\documentclass[a4paper, USenglish, cleveref, autoref, thm-restate]{lipics-v2021}
\pdfoutput=1 %uncomment to ensure pdflatex processing (mandatatory e.g. to submit to arXiv)
\hideLIPIcs  %uncomment to remove references to LIPIcs series (logo, DOI, ...), e.g. when preparing a pre-final version to be uploaded to arXiv or another public repository

\graphicspath{{./figure/}}

\title{Minimum Sum Coloring with Bundles in Trees and Bipartite Graphs} %TODO Please add

\author{Takehiro Ito}{Graduate School of Information Sciences, Tohoku University, Japan}{takehiro@tohoku.ac.jp}{https://orcid.org/0000-0002-9912-6898}{JSPS KAKENHI Grant Numbers JP24H00686, JP24H00690}
\author{Naonori Kakimura}{Faculty of Science and Technology, Keio University, Japan}{kakimura@math.keio.ac.jp}{https://orcid.org/0000-0002-3918-3479}{JSPS KAKENHI Grant Numbers 21H03397, 23K21646, JP22H05001}
\author{Naoyuki Kamiyama}{Institute of Mathematics for Industry, Kyushu University, Japan}{kamiyama@imi.kyushu-u.ac.jp}{https://orcid.org/0000-0002-7712-2730}{JSPS KAKENHI Grant Number JP24K14825}
\author{Yusuke Kobayashi}{Research Institute for Mathematical Sciences, Kyoto University, Japan}{yusuke@kurims.kyoto-u.ac.jp}{https://orcid.org/0000-0001-9478-7307}{JSPS KAKENHI Grant Numbers JP22H05001,  24K02901}
\author{Yoshio Okamoto}{Graduate School of Informatics and Engineering, The University of Electro-Communications, Japan}{okamotoy@uec.ac.jp}{https://orcid.org/0000-0002-9826-7074}{JSPS KAKENHI Grant Number JP23K10982}

\authorrunning{T. Ito, N. Kakimura, N. Kamiyama, Y. Kobayashi, and Y. Okamoto} %TODO mandatory. First: Use abbreviated first/middle names. Second (only in severe cases): Use first author plus 'et al.'

\Copyright{Takehiro Ito, Naonori Kakimura, Naoyuki Kamiyama, Yusuke Kobayashi, and Yoshio Okamoto} %TODO mandatory, please use full first names. LIPIcs license is "CC-BY";  http://creativecommons.org/licenses/by/3.0/


 %TODO mandatory; please add comma-separated list of keywords

\category{} %optional, e.g. invited paper

\relatedversion{} %optional, e.g. full version hosted on arXiv, HAL, or other respository/website
\acknowledgements{An extended abstract of this paper will appear in Proceedings of 36th International Symposium on Algorithms and Computation (ISAAC 2025). We are grateful to anonymous reviewers of ISAAC 2025 for their helpful comments.}%optional

\nolinenumbers %uncomment to disable line numbering

\EventEditors{John Q. Open and Joan R. Access}
\EventNoEds{2}
\EventLongTitle{42nd Conference on Very Important Topics (CVIT 2016)}
\EventShortTitle{CVIT 2016}
\EventAcronym{CVIT}
\EventYear{2016}
\EventDate{December 24--27, 2016}
\EventLocation{Little Whinging, United Kingdom}
\EventLogo{}
\SeriesVolume{42}
\ArticleNo{23}
\usepackage{xspace}
\newcommand{\cost}{\mathop{\mathrm{cost}}}
\newcommand{\mscwb}{\textsc{Minimum Sum Coloring with Bundles}\xspace}

\newcommand{\tw}{\mathop{\mathsf{tw}}}

\usepackage{comment}

\usepackage{mathtools}
\usepackage{framed}
\usepackage{booktabs}
\allowdisplaybreaks

\begin{document}

%\maketitle

\maketitle

%TODO mandatory: add short abstract of the document
\begin{abstract}
The minimum sum coloring problem with bundles was introduced by Darbouy and Friggstad (SWAT 2024) as a common generalization of the minimum coloring problem and the minimum sum coloring problem.
During their presentation, the following open problem was raised: whether the minimum sum coloring problem with bundles could be solved in polynomial time for trees.
We answer their question in the negative by proving that the minimum sum coloring problem with bundles is NP-hard even for paths.
We complement this hardness by providing algorithms of the following types.
First, we provide a fixed-parameter algorithm for trees when the number of bundles is a parameter; this can be extended to graphs of bounded treewidth.
Second, we provide a polynomial-time algorithm for trees when bundles form a partition of the vertex set and the difference between the number of vertices and the number of bundles is constant.
Third, we provide a polynomial-time algorithm for trees when bundles form a partition of the vertex set and each bundle induces a connected subgraph.
We further show that for bipartite graphs, the problem with weights is NP-hard even when the number of bundles is at least three, but is polynomial-time solvable when the number of bundles is at most two.
The threshold shifts to three versus four for the problem without weights.
\end{abstract}

%\clearpage 

\section{Introduction}

The minimum sum coloring problem with bundles was introduced by Darbouy and Friggstad~\cite{DBLP:conf/swat/DarbouyF24} as a common generalization of the minimum coloring problem and the minimum sum coloring problem.
For the minimum sum coloring problem with bundles, we are given an undirected graph $G$ and a family of vertex subsets of $G$, called bundles.
Then, we want to find a (proper) coloring of $G$ by positive integers such that the sum of the maximum colors in bundles is minimized (a precise definition will be given in the next section).
Figure~\ref{fig:mscwb_example} (Left) shows an example.
The minimum coloring problem corresponds to the case where the vertex set itself is a unique bundle; the minimum sum coloring problem corresponds to the case where each vertex forms a singleton bundle.

\begin{figure}[t]
    \centering
    \includegraphics{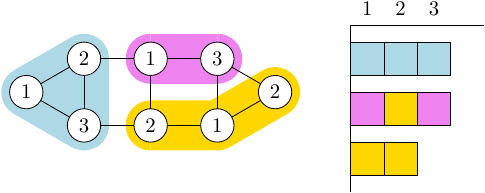}
    \caption{(Left) An instance of the minimum sum coloring problem with bundles. Bundles are distinguished by color. In this example, bundles are disjoint, but they do not have to be disjoint in general. Numbers on vertices represent the colors assigned to them. The maximum color is $3$ for the violet bundle, $3$ for the blue bundle, and $2$ for the orange bundle; their sum is $8$. (Right) Scheduling with conflicts. This is a Gantt chart, where each unit-time job corresponds to a square and the numbers correspond to time slots. Colors show bundles. The graph on the left is a conflict graph over the jobs, and the minimum sum coloring problem with bundles corresponds to minimizing the sum of makespans over bundles.}
    \label{fig:mscwb_example}
\end{figure}

The minimum sum coloring problem with bundles can be seen as a combinatorial model for the following scheduling problem.
We want to process $n$ jobs with unit processing time.
Among them, some pairs of jobs cannot be processed concurrently due to conflicts such as competing resources or precedence constraints.
Those conflicts are modeled by an undirected graph in which the vertices represent given jobs, and two vertices are joined by an edge if and only if the corresponding jobs are in conflict.
Furthermore, there are $\ell$ agents, and each of them is interested in the completion of some set of jobs.
Each of these sets of jobs is modeled as a bundle, and the agents want to minimize the completion time of the last job in each of their bundles.
We assume that there are sufficiently many machines with identical power.
Then, an optimal solution to the minimum sum coloring problem with bundles corresponds to a scheduling of jobs that minimizes the total (or average) completion time of the last jobs in the bundles.
Figure~\ref{fig:mscwb_example} (Right) shows the correspondence.
With this application in mind, it is natural to assume that the bundles form a partition of the vertex set.
However, some of our results can be applied even when the bundles do not form a partition.

The minimum sum coloring problem with bundles is a computationally hard problem.
This is expected since the minimum coloring problem is already NP-hard~\cite{DBLP:conf/coco/Karp72}.
However, it is known that the minimum coloring problem and the minimum sum coloring problem can be solved in linear time for trees~\cite{DBLP:conf/acm/KubickaS89}.
Darbouy and Friggstad~\cite{DBLP:conf/swat/DarbouyF24}, in their presentation at SWAT 2024, asked whether the minimum sum coloring problem with bundles can be solved in polynomial time for trees.

The first main contribution of this paper is to prove that the minimum sum coloring problem with bundles is NP-hard even when the graph is a path and bundles form a partition of the vertex set.
This answers the open question by Darbouy and Friggstad~\cite{DBLP:conf/swat/DarbouyF24} mentioned above since paths are trees.
On the other hand, we already know from the literature that, if the bundles form a partition, the problem is linear-time solvable when the number of bundles is one (corresponding to the minimum coloring problem) or when it is equal to the number of vertices (corresponding to the minimum sum coloring problem).
We also observe that in our reduction, the subgraph induced by a bundle can be disconnected.
Therefore, we wonder what happens if the number of bundles is small or large, or if each bundle induces a connected subgraph.

The second contribution of this paper is to provide a fixed-parameter algorithm for the minimum sum coloring problem with bundles in trees when the number of bundles is a parameter.
The algorithm also works when each bundle is associated with a non-negative weight that acts as a coefficient in the sum.
The algorithm is based on a combinatorial lemma that states that the number of colors in an optimal coloring is bounded by the chromatic number of the graph multiplied by the number of bundles, which holds for general graphs.

The third contribution is to provide a polynomial-time algorithm for the minimum sum coloring problem with bundles in trees when bundles form a partition and the difference between the number of vertices and the number of bundles is small.
It should be noted that this is not a fixed-parameter algorithm with respect to that difference.

The fourth contribution is to provide a polynomial-time algorithm for the minimum sum coloring problem with bundles in trees when the family of bundles forms a partition of the vertex set and each bundle induces a connected subgraph.
Note that this algorithm does not require the number of bundles to be bounded.
We also show that the minimum sum coloring problem with bundles in paths can be solved in polynomial time when each bundle induces a connected subgraph (even if the bundle family does not necessarily form a partition).

Next, we turn our attention to bipartite graphs.
Since for bipartite graphs, the minimum coloring (i.e., when the number of bundles is one) is easy and the minimum sum coloring (i.e., when the number of bundles is equal to the number of vertices) is hard~\cite{DBLP:journals/jal/Bar-NoyK98}, we investigate how the number of bundles affects the computational complexity of the minimum sum coloring problem with bundles.

As the fifth contribution, we prove that for bipartite graphs, the minimum sum coloring problem with bundles can be solved in polynomial time when the number of bundles is at most three, and it is NP-hard when the number of bundles is at least four.
For the weighted case, the problem can be solved in polynomial time when the number of bundles is at most two, and it is NP-hard when the number of bundles is at least three.

Our results are summarized in \Cref{tab:summary}.

%\textbf{\textcolor{red}{TODO}: Add a table for the result summary.}

\begin{table}[t]
    \centering
    \caption{Summary of the results in this paper. Here, $n$ is the number of vertices.}
    \begin{tabular}{llllll}
    \toprule
         &  \# bundles & Type of bundles & Weights & Results & Reference\\
    \midrule
         Trees
         & Unbounded & Independent partition & Uniform & NP-complete & Theorem \ref{thm:hardpath}\\
         & Parameter & General & General & FPT & Theorem \ref{thm:fpt}\\
         & $n-\text{Parameter}$ & Partition & General & XP & Theorem \ref{thm:xp}\\
         & Unbounded & Connected partition & General & P & Theorem \ref{thm:connected_tree}\\
    \midrule
         Bipartite 
         & $\leq 3$ & General & Uniform & P & \Cref{thm:poly-bip-unif-3}\\
         & $\geq 4$ & General & Uniform & NP-complete & \Cref{thm:BipartiteNPhard}\\
         & $\leq 2$ & General & General & P & \Cref{thm:poly-bip-gen-2}\\
         & $\geq 3$ & General & General & NP-complete & \Cref{thm:BipartiteNPhard2}\\
    \bottomrule
    \end{tabular}
    \label{tab:summary}
\end{table}

\subsection*{Related work}

As mentioned in the introduction, the minimum sum coloring problem with bundles was introduced by Darbouy and Friggstad~\cite{DBLP:conf/swat/DarbouyF24} as a common generalization of the minimum coloring problem and the minimum sum coloring problem.

The minimum coloring problem has been one of the central topics in graph theory, graph algorithms, combinatorial optimization, and computational complexity theory.
It is known that the minimum coloring problem is NP-complete~\cite{DBLP:conf/coco/Karp72},\footnote{%
Whenever we talk about NP-completeness of a minimization problem, we consider a canonical decision version of the problem in which we are also given an upper bound $C$ of the optimal value, and want to decide whether the optimal value is at most $C$.%
} even for $4$-regular planar graphs~\cite{DBLP:journals/dm/Dailey80}.
For several classes of graphs, the minimum coloring problem can be solved in polynomial time.
Most notably, for perfect graphs, a polynomial-time algorithm is known that is based on semidefinite programming relaxation~\cite{GLSbook}.
The class of perfect graphs contains bipartite graphs, interval graphs, and chordal graphs, for which the minimum coloring problem can be solved in linear time.

The minimum sum coloring problem was introduced by Kubicka and Schwenk~\cite{DBLP:conf/acm/KubickaS89}, where the NP-completeness of the problem was established.
The problem is also known to be NP-complete even for bipartite graphs~\cite{DBLP:journals/jal/Bar-NoyK98} and interval graphs~\cite{DBLP:journals/siamcomp/Szkaliczki99, DBLP:journals/orl/Marx05} (and hence for chordal graphs).
We recommend a survey by Halld\'orsson and Kortsarz~\cite{DBLP:books/tf/18/HalldorssonK18} on approximability for the minimum sum coloring.

The minimum coloring problem and the minimum sum coloring problem are closely related to scheduling with conflict or mutual exclusion scheduling.
In a basic version of scheduling with conflict, 
we are given a set of $n$ jobs with unit processing time, a set of $m$ machines, and a conflict graph $G$ on the set of jobs.
Then, we need to assign each job to one of the machines and process the jobs along the timeline in such a way that two jobs are not processed at the same time slot if they are adjacent in $G$.
The objective is to minimize either the makespan or the total completion time.
The minimum coloring problem corresponds to the case where $m \ge n$ and the objective is makespan minimization;
the minimum sum coloring problem corresponds to the case where $m \ge n$ and the objective is total completion time minimization.
Baker and Coffman Jr.~\cite{DBLP:journals/tcs/BakerC96} proved that the problem of scheduling with conflict can be solved in polynomial time when $m=2$ and is NP-hard when $m \geq 3$.
For more recent and related results, see a paper by Even et al.~\cite{DBLP:journals/scheduling/EvenHKR09}.
We note that there are versions of scheduling with conflicts where we cannot process two jobs in conflict on the same machine~\cite{DBLP:journals/dam/BodlaenderJW94}.

\section{Preliminaries}

For a positive integer $k \in \mathbb{Z}_{>0}$, we use the notation $[k] = \{1, 2, \dots , k\}$. 

A \emph{graph} $G$ is a pair of its \emph{vertex set} $V(G)$ and its \emph{edge set} $E(G)$.
Each edge $e \in E(G)$ is an unordered pair $\{u,v\}$ of vertices of $G$, where $u$ and $v$ are \emph{endpoints} of the edge $e$.
We follow the basic terminology for graphs
from Korte and Vygen \cite{kortevygen}.
%All graphs in this paper are finite, undirected, and simple (i.e., having no self-loops or multiple edges).
A graph $G=(V, E)$ is \emph{bipartite} if the vertex set $V$ can be partitioned into two disjoint sets $A$ and $B$ in such a way that every edge $e \in E$ has one endpoint in $A$ and the other endpoint in $B$.
A graph $G$ is a \emph{tree} if it is connected and has no cycles.
For $v \in V$, let $\delta_G(v)$ denote the set of all edges incident to $v$. 

Let $G=(V, E)$ be a graph. 
A (proper) \emph{coloring} of $G$ is a map $c\colon V \to \mathbb{Z}_{>0}$ such that $c(u) \neq c(v)$ for every edge $uv \in E$.
For a vertex subset $S \subseteq V$, 
an \emph{$S$-partial coloring} of $G$ is a map $c\colon V \to \mathbb{Z}_{\geq 0}$ such that
$c(v) > 0$ for $v \in S$, $c(v) = 0$ for $v \in V \setminus S$, and 
$c(u) \neq c(v)$ for every edge $uv \in E$ with $u, v \in S$.

The problem \mscwb is defined as follows.

\begin{framed}
\begin{description}
\item[Problem:] \mscwb
\item[Input:] An undirected graph $G=(V, E)$, 
a family $\mathcal{B} = \{B_1, B_2, \dots, B_\ell\}$ of non-empty subsets of $V$,
a weight function $w \colon \mathcal{B} \to \mathbb{Z}_{> 0}$,
and
a positive integer $C \in \mathbb{Z}_{> 0}$.
\item[Question:] Determine whether there exists a (proper) coloring
$c\colon V \to \mathbb{Z}_{>0}$ such that
\[
\cost(c) := \sum_{j=1}^{\ell} w(B_j) \max\{ c(v) \mid v \in B_j\} \leq C.
\]
\end{description}
\end{framed}
For an instance $I=(G, \mathcal{B}, w, C)$ of \mscwb,
we call a coloring $c$ of $G$ \emph{optimal} if $\cost(c) \leq \cost(c')$ for all colorings $c'$ of $G$.

Each member of $\mathcal{B}$ is called a \emph{bundle}.
Throughout the paper, we assume that each vertex $v \in V$ is contained in at least one bundle; otherwise, we can remove $v$ from the instance.  
The family $\mathcal{B}$ is a \emph{partition} if for each vertex $v \in V$ there exists a unique bundle $B_j$ such that $v \in B_j$.
A bundle $B$ is \emph{connected} if the induced subgraph $G[B]$ is connected.
A bundle $B$ is \emph{independent} if the induced subgraph $G[B]$ has no edges.
The family $\mathcal{B}$ is called \emph{connected} (and \emph{independent}) if all its members are connected (and independent, respectively).

The following observation will be used later explicitly or implicitly.
\begin{observation}
\label{obs:maxdeg}
For any graph $G = (V, E)$ in \mscwb, there exists an optimal coloring in which the color of every vertex $v$ is at most $|\delta_G(v)|+1$.
\end{observation}
\begin{proof}
    If the color of $v$ exceeds $|\delta_G(v)|+1$, then one of the colors in $\{1,2,\dots,|\delta_G(v)|+1\}$ is missing in the neighborhood of $v$, and $v$ can be recolored with that missing color without increasing the cost.
\end{proof}

\section{Trees}\label{sec:trees}

\subsection{NP-Completeness for Paths}

As an intermediate step, we first prove that \mscwb is NP-complete for perfect matchings.
Here, a \emph{perfect matching} is a set of edges in which no pair of edges shares a common endpoint.

\begin{theorem}
\label{thm:hardmatching}
\mscwb is NP-complete even when 
the edge set of the input graph forms a perfect matching, 
$\mathcal{B}$ is an independent partition of the vertex set, and 
$w(B) = 1$ for every $B \in \mathcal{B}$. 
\end{theorem}

\begin{proof}
It is easy to see that \mscwb is in NP\@.
To prove NP-hardness, we reduce \textsc{Independent Set} to \mscwb. 
Here, for an undirected graph $G=(V, E)$, we say that $S \subseteq V$ is \emph{independent} if 
$G$ has no edge connecting the vertices in $S$.
\begin{description}
\item[Problem:] \textsc{Independent Set}
\item[Input:] An undirected graph $G=(V, E)$
and
a positive integer $k \in \mathbb{Z}_{> 0}$.
\item[Question:] Determine whether there exists an independent set $S \subseteq V$ with $|S| \geq k$.
\end{description}
It is well-known that \textsc{Independent Set} is NP-complete (see~\cite{DBLP:conf/coco/Karp72}). 

Suppose that we are given an instance of \textsc{Independent Set} that 
consists of a graph $G=(V, E)$ and a positive integer $k$. 
We construct a new graph $G' = (V', E')$ by replacing each vertex $v \in V$ with $|\delta_G(v)|$ new vertices 
so that each new vertex has exactly one incident edge. 
Formally, $G'$ is defined as follows: 
\begin{align*}
V' &= \{p_{(v, e)} \mid v \in V,\ e \in \delta_G(v)\}, & 
E' &= \{ p_{(u, e)}p_{(v, e)} \mid e = uv \in E\}. 
\end{align*}
Let $B_v := \{ p_{(v, e)} \mid e \in \delta_G(v)\}$ for each $v \in V$, and define  
$\mathcal{B} = \{B_v \mid v \in V\}$. 
Then, it is easy to see that $E'$ forms a perfect matching in $G'$ and 
$\mathcal{B}$ is a partition of $V'$. 
Let $w(B) = 1$ for any $B \in \mathcal{B}$, and let $C = 2 |V| - k$. 
This defines an instance 
$(G', \mathcal{B}, w, C)$ of \mscwb. 
See \Cref{fig:reduction_perfmatch1}.

To show the validity of this reduction, it suffices to show that 
$G$ has an independent set of size at least $k$ if and only if 
there exists a coloring $c$ of $G'$ such that 
$\sum_{v \in V} \max \{c (p) \mid p \in B_v \} \le C$.

We first show the sufficiency (``if'' part). 
Suppose that $G'$ has a coloring 
$c \colon V' \to \mathbb{Z}_{> 0}$ such that 
$\cost(c) \le C=2|V|-k$. 
Let $S := \{v \in V \mid  \max \{c (p) \mid p \in B_v \} = 1\}$. 
Then, since 
\[
C \ge \cost(c) = \sum_{v \in V} \max \{c (p) \mid p \in B_v \}
\ge 2(|V| - |S|) + |S|, 
\]
we obtain $|S| \ge k$. 
For any edge $e = uv \in E$, 
since $c(p_{(u, e)}) \neq c(p_{(v, e)})$ as $c$ is a coloring, 
at least one of $c(p_{(u, e)}) \ge 2$ and $c(p_{(v, e)}) \ge 2$ holds, 
which implies that at least one of $u$ and $v$ is in $V \setminus S$. 
Therefore, $S$ is an independent set in $G$, 
which shows the sufficiency.  

We next show the necessity (``only if'' part). 
Suppose that $G$ has an independent set $S \subseteq V$ with $|S| \ge k$. 
Let $c \colon V' \to \mathbb{Z}_{> 0}$ be a coloring of $G'$ such that 
\begin{itemize}
\item
for each edge $e = uv \in E$, 
one of $c(p_{(u, e)})$ and $c(p_{(v, e)})$ is $1$ and 
the other is $2$, and 
\item
for any $v \in S$ and any $e \in \delta_G(v)$, it holds that $c(p_{(v, e)}) = 1$. 
\end{itemize}
Note that such a coloring $c$ exists, because $S$ is an independent set in $G$. 
Since $\max \{c (p) \mid p \in B_v \} = 1$ for $v \in S$ and 
$\max \{c (p) \mid p \in B_v \} \le 2$ for $v \in V \setminus S$, 
we obtain 
\[
\cost(c) = \sum_{v \in V} \max \{c (p) \mid p \in B_v \} \le 2 |V| - |S| \le C, 
\]
which shows the necessity. 

Therefore, the reduction is valid, and hence 
\mscwb is NP-hard even when 
the edge set forms a perfect matching, 
$\mathcal{B}$ is an independent partition of the vertex set, and  
$w(B) = 1$ for every $B \in \mathcal{B}$. 
\end{proof}

\begin{figure}[t]
    \centering
    \includegraphics[scale=0.8]{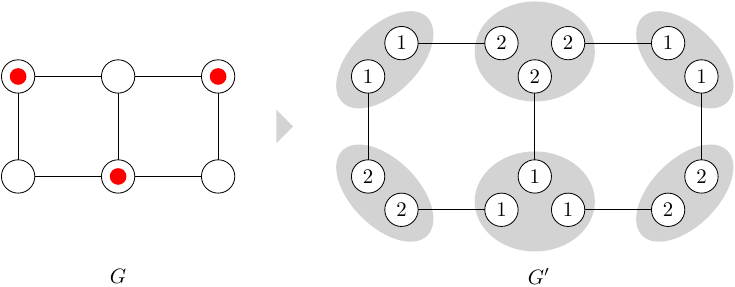}
    \caption{Reduction when a graph forms a perfect matching.}
    \label{fig:reduction_perfmatch1}
\end{figure}

We then prove the NP-hardness of \mscwb on paths by reducing \mscwb with the conditions in Theorem~\ref{thm:hardmatching}.
Roughly speaking, the reduction involves creating four copies of the instance described in Theorem~\ref{thm:hardmatching} and carefully connecting their edges to form a single path. 

\begin{theorem}
\label{thm:hardpath}
\mscwb is NP-complete even when 
$G$ is a path,  
$\mathcal{B}$ is an independent partition of the vertex set, and  
$w(B) = 1$ for every $B \in \mathcal{B}$.  
\end{theorem}
\begin{proof}
It is easy to see that \mscwb is in NP\@. 
To prove NP-hardness for the setting in Theorem~\ref{thm:hardpath}, 
we reduce \mscwb with the conditions in Theorem~\ref{thm:hardmatching}. 

Suppose that $(G, \mathcal{B}, w, C)$ is an instance of \mscwb with the conditions in Theorem~\ref{thm:hardmatching}. 
Since the edge set of $G=(V, E)$ forms a perfect matching, 
we denote $V = \{v_1, v_2, \dots , v_{2n-1}, v_{2n}\}$ and 
$E = \{v_1 v_2, v_3 v_4, \dots , v_{2n-1} v_{2n}\}$. 

We create four disjoint copies of this instance. 
For $i \in [4]$, the objects in the $i$-th copy are represented by symbols with the superscript $i$, 
e.g., the $i$-th copy of $G$ is denoted by $G^i$. 
We construct a path $G' = (V', E')$ by introducing new vertices $x^i_j$ and $y^i$ as follows (see \Cref{fig:reduction_path1}):
\begin{align*}
V' &= \left(\bigcup_{i \in [4]} V^i \right) \cup \left\{ x^i_j \mid i \in [4], j\in [n-1]\right\} \cup \left\{y^1, y^2, y^3\right\}, \\
E' &= \left(\bigcup_{i \in [4]} E^i \right) \cup \left\{ v^i_{2j} x^i_j, x^i_j v^i_{2j+1} \mid i \in [4], j\in [n-1]\right\} \cup \left\{ v^i_{2n} y^i, y^i v^{i+1}_{1} \mid i \in [3]\right\}. 
\end{align*}
Let $B_0 = \left\{ x^i_j \mid i \in [4], j\in [n-1]\right\} \cup \left\{y^1, y^2, y^3\right\}$ and 
let $\mathcal{B}' = \left(\bigcup_{i \in [4]} \mathcal{B}^i \right) \cup \{B_0\}$. 
Then, $G'$ is a path and $\mathcal{B}'$ is an independent partition of $V'$. 
Set $w'(B) = 1$ for each bundle $B \in \mathcal{B}'$ and let $C' = 4C + 3$. 
This defines an instance 
$(G', \mathcal{B}', w', C')$ of \mscwb with the conditions in Theorem~\ref{thm:hardpath}. 

To show the validity of this reduction, it suffices to show that 
$(G, \mathcal{B}, w, C)$ is a yes-instance if and only if 
$(G', \mathcal{B}', w', C')$ is a yes-instance. 

We first show the necessity (`only if' part). 
Suppose that $G$ has a coloring 
$c \colon V \to \mathbb{Z}_{> 0}$ such that 
$\sum_{B \in \mathcal{B}} \max  \{c (v) \mid v \in B \} \le C$. 
Since each vertex of $G$ has degree one, 
we may assume that $c(v) \in \{1, 2\}$ for $v \in V$ by Observation~\ref{obs:maxdeg}. 
Define $c' \colon V' \to \mathbb{Z}_{> 0}$ as follows: 
\begin{itemize}
\item
for $v \in V$ and $i \in [4]$, let $c'(v^i) = c(v)$, 
where $v^i \in V^i$ is the $i$-th copy of $v$, and 
\item
let $c'(v') = 3$ for $v' \in \{ x^i_j \mid i \in [4], j\in [n-1]\} \cup \{y^1, y^2, y^3\}$.  
\end{itemize}
Then, $c'$ is a coloring of $G'$ such that 
$\sum_{B' \in \mathcal{B}'} \max \{c' (v') \mid v' \in B' \} \le 4C + 3 = C'$, 
which shows that $(G', \mathcal{B}', w', C')$ is a yes-instance.

We next show the sufficiency (`if' part). 
Suppose that $G'$ has a coloring 
$c' \colon V' \to \mathbb{Z}_{> 0}$ such that 
$\sum_{B' \in \mathcal{B}'} \max \{c' (v') \mid v' \in B' \} \le C'$. 
Since 
\begin{align*}
\sum_{i \in [4]} \sum_{B' \in \mathcal{B}^i} \max \{c' (v') \mid v' \in B' \}
\le 
\sum_{B' \in \mathcal{B}'} \max \{c' (v') \mid v' \in B' \} \le C' = 4C + 3,  
\end{align*}
there exists some $i \in [4]$ such that 
\[
\sum_{B' \in \mathcal{B}^i} \max \{c' (v') \mid v' \in B'\} \le C. 
\]
Therefore, the restriction of $c'$ to $V^i$ is a solution to $(G^i, \mathcal{B}^i, w^i, C^i)$, 
which is the $i$-th copy of $(G, \mathcal{B}, w, C)$. 
This implies that $(G, \mathcal{B}, w, C)$ is a yes-instance.  

Therefore, the reduction is valid, and hence 
\mscwb is NP-hard even when 
the graph is a path, 
$\mathcal{B}$ is an independent partition of the vertex set, and  
$w(B) = 1$ for every $B \in \mathcal{B}$. 
\end{proof}

\begin{figure}[t]
    \centering
    \includegraphics{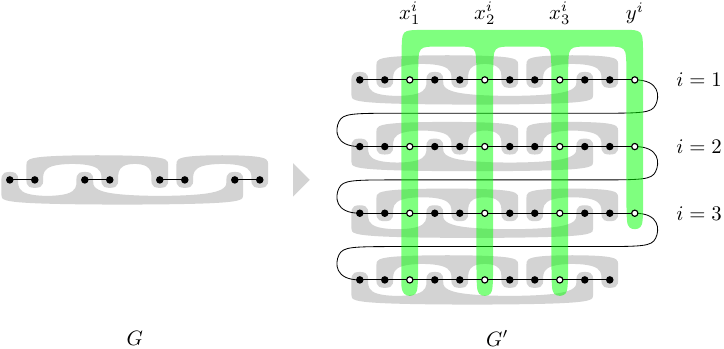}
    \caption{Reduction when a graph is a path.}
    \label{fig:reduction_path1}
\end{figure}

\subsection{Fixed-Parameter Tractability when $|\mathcal{B}|$ is a Parameter}

In the reduction of the previous section, the number of bundles is unbounded.
As a positive algorithmic result that is in contrast to \Cref{thm:hardpath}, we prove that \mscwb in graphs of bounded treewidth can be solved in fixed-parameter tractable time when the number of bundles and the treewidth are parameters.
See Appendix for an introduction to treewidths and Courcelle's theorem that we use for our proof.
%See a textbook (e.g.\ \cite{DBLP:books/sp/CyganFKLMPPS15}) for an introduction to treewidths and Courcelle's theorem.

First, we prove the lemma that bounds the number of colors in any optimal coloring.
For a graph $G=(V,E)$, we denote by $\chi(G)$ the \emph{chromatic number} of $G$, which is defined as the minimum of $\max_{v \in V}c(v)$ over all (proper) colorings $c$ of $G$.

\begin{lemma}
\label{lem:ub_by_chi}
    Let $I=(G=(V,E),\mathcal{B},w,C)$ be an instance of \mscwb such that each vertex is contained in some bundle. 
    Then, every optimal coloring $c\colon V \to \mathbb{Z}_{>0}$ for $I$ satisfies $c(v) \leq \chi(G)|\mathcal{B}|$ for all $v \in V$.
\end{lemma}
\begin{proof}
    The proof proceeds by induction on $|\mathcal{B}|$.
    Consider the case where $|\mathcal{B}| = 1$ and $\mathcal{B} = \{B\}$.
    Note that $B = V$ as each vertex is contained in some bundle. 
    Let $c$ be a minimum coloring, i.e., a coloring that attains $\chi(G)$.
    Then, $\max_{v \in B}c(v) \leq \max_{v \in V}c(v) = \chi(G) = \chi(G)|\mathcal{B}|$.

    Now, let $\ell \geq 1$ be a positive integer, and assume that the lemma holds when the number of bundles is $\ell$.
    Then, we prove that the lemma holds when the number of bundles is $\ell+1$.
    For the sake of contradiction, suppose that there exists an optimal coloring $c^*$ such that $\max_{v \in B}c^*(v) > \chi(G)(\ell+1)$ for some $B \in \mathcal{B}$.
    Fix such a member $B$, and let $U$ be the set of vertices of $G$ that belong to $B$, but not to other bundles. Then, consider the instance $I'=(G-U,\mathcal{B}\setminus \{B\}, w|_{V\setminus U}, C)$, where $w|_{V \setminus U}$ is the restriction of $w$ on $V \setminus U$.
    By the induction hypothesis, there exists an optimal coloring $c'\colon V\setminus U \to \mathbb{Z}_{>0}$ for $I'$ such that $\max_{v \in V \setminus U}c'(v) \leq \chi(G-U)\ell \leq \chi(G)\ell$.
    We now extend $c'$ to a coloring $c$ of $G$ by giving colors in $\{\chi(G)\ell+1, \dots, \chi(G)\ell+\chi(G)\}$ to the vertices of $U$.
    This is possible since $\chi(G[U]) \leq \chi(G)$.
    Then,
    \begin{alignat*}{3}
        \cost(c)
        &= \sum_{B' \in \mathcal{B}\setminus \{B\}}w(B')\max_{v \in B'}c(v) + w(B)\max_{v \in B}c(v)\\
        &\leq \sum_{B' \in \mathcal{B}\setminus\{B\}}w(B')\max_{v \in B'}c'(v) + w(B)\chi(G)(\ell+1)\\
        &< \sum_{B' \in \mathcal{B}\setminus\{B\}}w(B')\max_{v \in B'}c^*(v) +w(B)\max_{v \in B}c^*(v)
         &= \cost(c^*).
    \end{alignat*}
    This contradicts the optimality of $c^*$.
\end{proof}

We now focus on the case where the input graph $G$ has treewidth at most $\tw$.
Note that $\chi(G) \leq \tw+1$, because the degeneracy of $G$ is at most $\tw$. 
Thus, the number of colors in any optimal coloring is at most $(\tw+1)|\mathcal{B}|$ by \Cref{lem:ub_by_chi}.

To solve \mscwb for graphs $G=(V,E)$ with treewidth at most $\tw$, we employ Courcelle's theorem~\cite{DBLP:journals/jal/ArnborgLS91,DBLP:journals/iandc/Courcelle90,DBLP:journals/ita/Courcelle92}.
To this end, we introduce the following auxiliary decision problem.
In addition to the input $(G, \mathcal{B}, w, C)$ to \mscwb, where $\mathcal{B} = \{B_1, B_2, \dots, B_\ell\}$, we also take a sequence $\mathbf{k} = (k_1, k_2, \dots, k_\ell)$ of integers.
Then, we want to decide whether there exists a coloring $c\colon V\to \mathbb{Z}_{>0}$ such that $\max_{v \in B_j}c(v) \leq k_j$ for all $j \in [\ell]$.
If we have a solution to this auxiliary problem for every $\mathbf{k} \in [(\tw+1)\ell]^\ell$, then the optimal cost for the instance $(G, \mathcal{B}, w, C)$ can be derived as
\begin{equation}
\min_{\mathbf{k}}\left\{ \sum_{j=1}^{\ell} w(B_j)k_j \,\middle|\, \text{the answer to the instance } (G,\mathcal{B},w,C,\mathbf{k}) \text{ is yes} \right\}.\label{eq:fpt}
\end{equation}

We now focus on solving the auxiliary problem above.
To solve the auxiliary problem using Courcelle's theorem~\cite{DBLP:journals/jal/ArnborgLS91,DBLP:journals/iandc/Courcelle90,DBLP:journals/ita/Courcelle92}, it is sufficient to write down an MSO formula that expresses the decision.
An instance is given as $(G,\mathcal{B},w,C,\mathbf{k})$, and we want to decide whether there exists a coloring $c$ such that $\max_{v \in B_j}c(v) \leq k_j$ for all $j \in [\ell]$.
By \Cref{lem:ub_by_chi}, we can assume that our coloring $c$ uses colors in $[(\tw+1)\ell]$.
For brevity, let $p=(\tw+1)\ell$.
To express conditions with an MSO formula, we regard a coloring $c$ as a partition $\{C_1,C_2,\dots,C_p\}$ of the vertex set $V$ into independent sets, where $C_i = \{v \in V \mid c(v) = i\}$.
Note that in such a partition a set $C_i$ can be empty.

Then, the decision can be expressed by the following MSO formula:
\begin{align}
    \exists\, C_1,C_2,\dots,C_p \subseteq V\colon&
    \bigwedge_{i=1}^{p}\bigwedge_{i'=i+1}^{p}
    \forall\, v\in V\colon \neg (v \in C_i \wedge v \in C_{i'})\label{eq:mso-disjoint}\\
    &\wedge \forall \ 
 v\in V\colon \bigvee_{i=1}^{p} (v\in C_i)\label{eq:mso-cover}\\
    &\wedge \bigwedge_{i=1}^{p}
    \forall\, u,v \in V\colon \left(u\in C_i \wedge v \in C_i \rightarrow \neg (uv \in E)\right)\label{eq:mso-independent}\\
    &\wedge \bigwedge_{j=1}^{\ell}\forall\, v \in V\colon \left(v \in B_j \rightarrow \bigvee_{i=1}^{k_j} (v\in C_i)\right).\label{eq:mso-bundle}
\end{align}
In the formula, the conjunction of Formulae (\ref{eq:mso-disjoint}) and (\ref{eq:mso-cover}) represents that $\{C_1, C_2, \dots, C_p\}$ is a partition of $V$.
Formula (\ref{eq:mso-independent}) represents that each $C_i$ is independent.
Formula (\ref{eq:mso-bundle}) represents that the maximum color in bundle $B_j$ is at most $k_j$ for each $j \in [\ell]$.
Thus, the formula correctly represents the conditions in the auxiliary problem.

By Courcelle's theorem, this auxiliary problem can be solved in $O(f(\tw,\ell)|V|)$ time for some computable function $f$.
To evaluate the minimum in (\ref{eq:fpt}), we solve the auxiliary problem for every choice of $\mathbf{k} \in [(\tw+1)\ell]^\ell$, thus at most $((\tw+1)\ell)^\ell$ times.
Let $g(\tw)|V|$ be the running time of Bodlaender's algorithm~\cite{DBLP:journals/siamcomp/Bodlaender96} for computing a tree decomposition of width $\tw$.
Then, the overall running time is $g(\tw)|V| + ((\tw+1)\ell)^\ell \cdot O( f(\tw,\ell)|V|)$.
Setting $h(\tw,\ell) = g(\tw) +  ((\tw+1)\ell)^\ell f(\tw,\ell)$, we obtain the following theorem.

\begin{theorem}
\label{thm:fpt}
\mscwb can be solved in $O(h(\tw,\ell)|V|)$ time for some computable function $h\colon \mathbb{Z}_{\geq 0}\times \mathbb{Z}_{\geq 0} \to \mathbb{Z}_{\geq 0}$, where $\tw$ is the treewidth of $G$ and $\ell$ is the number of bundles.
\qed
\end{theorem}

\subsection{An XP-Algorithm when $|V|-|\mathcal{B}|$ is a Parameter}

We consider the case where $G$ is a tree, the bundle family is a partition of the vertex set, and the number of bundles is large.
More specifically, if $n$ is the number of vertices and $\ell$ is the number of bundles that form a partition of the vertex set, we treat $n-\ell$ as a parameter.

Let $B \in \mathcal{B}$ be a bundle.
We call $B$ a \emph{singleton bundle} if $|B|=1$; otherwise, we call $B$ a \emph{non-singleton bundle}.

\begin{lemma}
\label{lem:nonsingletonbundles}
    Let $I=(G=(V,E),\mathcal{B},w,C)$ be an instance of \mscwb, where $\mathcal{B}$ is a partition of $V$.
    Then, the number of non-singleton bundles in $\mathcal{B}$ is at most $|V|-|\mathcal{B}|$.
\end{lemma}
\begin{proof}
    Let $t$ be the number of non-singleton bundles.
    Since $\mathcal{B}$ is a partition of $V$, we obtain $|V| \geq 2t + (|\mathcal{B}|-t) = |\mathcal{B}| + t$. Thus, $t \leq |V|-|\mathcal{B}|$.
\end{proof}

Let $I=(G=(V,E),\mathcal{B},w,C)$ be an instance of \mscwb, and assume that
$G$ is a tree, $\mathcal{B}$ is a partition of $V$.
Let $t$ be the number of non-singleton bundles in $\mathcal{B}$.
By \cref{lem:nonsingletonbundles}, $t \leq |V|-|\mathcal{B}| = n-\ell$.
Denote the non-singleton bundles in $\mathcal{B}$ by $B_1, B_2, \dots, B_t$, and
the set of vertices, each of which forms a singleton bundle, by $V_s$.

A primary idea of our algorithm is to fix upper bounds on the colors used in non-singleton bundles and to optimize the colors in singleton bundles with this restriction.
For a non-singleton bundle $B_j$, $j\in [t]$, let $k_j$ be an upper bound on the colors in $B_j$; namely, we seek a coloring $c$ such that $\max_{u \in B_j}c(u) \leq k_j$ for every $j \in [t]$.
For brevity, we denote $\mathbf{k} = (k_1,k_2,\dots, k_t)$.
By Lemma~\ref{lem:ub_by_chi}, it is enough to consider the situations with $k_j \leq 2\ell$ for all $j \in [t]$, and thus the number of possible choices for $\mathbf{k}$ is bounded by $(2\ell)^t$ from above.

We regard $G$ as a rooted tree with a root $r \in V$.
For each vertex $v \in V$, we denote by $V(v)$ the vertex set of the subtree of $G$ rooted at $v$.
For each vertex $v \in V$, each color $k$, and an upper-bound tuple $\mathbf{k}$,  
we define
\begin{gather*}    
f(v,k; \mathbf{k})
:=
\min\left\{ \sum_{u \in V_s} w(\{u\}) c(u)
\,\middle|\,
\begin{array}{l}
c \text{ is a $V(v)$-partial coloring of $G$},\\
c(v) = k,\\
\text{$\max_{u \in B_j}c(u) \leq k_j$ for all $j \in [t]$}\\
\end{array}
\right \}.
\end{gather*}
Then, the optimal value for the instance $I$ is read by
\[
\min_{\mathbf{k},k} \left(f(r,k;\mathbf{k}) + \sum_{j=1}^{t}k_j \right),
\]
where the minimum is taken over all possible choices of $\mathbf{k}=(k_1,k_2,\dots,k_t)$ and $k$.
Note that we only need to consider the value of $k$ in the interval $k \in [2\ell]$ by Lemma \ref{lem:ub_by_chi}.

We now establish a recursive formula for $f(v,k; \mathbf{k})$.
First, consider the case where $v$ is a leaf of $G$.
We have two subcases.
If $v$ forms a singleton bundle by itself, then
\[
f(v,k;\mathbf{k}) = w(\{v\})k.
\]
If $v$ belongs to a non-singleton bundle $B_j$ for some $j \in [t]$, which is unique since $\mathcal{B}$ forms a partition of $V$, then
\[
f(v,k;\mathbf{k})=
\begin{cases}
    0 & \text{if $k \leq k_j$},\\
    +\infty & \text{otherwise}.\\
\end{cases}
\]

Next, consider the case where $v$ is not a leaf of $G$.
Let $v_1, \dots, v_z$ be the children of $v$ in $G$.
We again have two subcases.
If $v$ forms a singleton bundle by itself, then
\begin{align*}
f(v,k;\mathbf{k})
&= w(\{v\})k + \min \left\{\sum_{y=1}^{z} f(v_y, k'_{y}; \mathbf{k})
\,\middle|\,
k'_{y} \neq k \text{ for all } y \in [z]
\right\}\\
&= w(\{v\})k + \sum_{y=1}^{z}\min \left\{ f(v_y, k'_{y}; \mathbf{k})
\,\middle|\,
k'_{y} \neq k
\right\}
\end{align*}
since the values $f(v_y,k'_{y};\mathbf{k})$ are independent from each other.
Similarly, if $v$ belongs to a non-singleton bundle $B_j$ for some $j \in [t]$, then
\[
f(v,k;\mathbf{k})
= 
\begin{cases}
    \displaystyle \sum_{y=1}^{z} \min \left\{f(v_y, k'_{y}; \mathbf{k})
\,\middle|\,
k'_{y} \neq k\right\} & \text{if $k \leq k_j$},\\
    +\infty & \text{otherwise}.
\end{cases}
\]

For each fixed $\mathbf{k}$, we compute $f(v,k;\mathbf{k})$ in a bottom-up manner from leaves to root.
When $v$ is not a leaf, the computation of $f(v,k;\mathbf{k})$ can be done in $O(z\ell)$ time 
for each fixed $k$
since we may suppose that $k'_{y} \leq 2\ell$ by Lemma \ref{lem:ub_by_chi}.
In total, for each fixed $\mathbf{k}$, the computation of $f(v,k;\mathbf{k})$ over all $v$ and $k$ can be done in $O(n\ell^2)$ time.
Hence, for varying $\mathbf{k}$, the running time of our algorithm is $O((2\ell)^t n\ell^2)$.
When $n - \ell$ is a parameter, this is an XP algorithm since $t \leq n-\ell$ by Lemma~\ref{lem:nonsingletonbundles}.

We summarize our findings in the following theorem.
\begin{theorem}
\label{thm:xp}
    \mscwb can be solved in $O((2\ell)^t n \ell^2)$ time when $G$ is a tree, $n$ is the number of vertices, $\ell$ is the number of bundles, and $t$ is the number of non-singleton bundles. 
\end{theorem}

\subsection{Polynomial-time Algorithm for the Connected Partition Case}

In this subsection, we consider \mscwb in trees when the bundle family is a connected partition, namely, when each bundle induces a connected subgraph and each vertex belongs to a unique bundle.

Let $I=(G=(V,E),\mathcal{B},w,C)$ be an instance of \mscwb, and assume that
$G$ is a tree, $\mathcal{B}$ is a connected partition of $V$.
Define $n \coloneqq |V|$. 
For each vertex $v \in V$, 
we denote by $B(v)$ a unique bundle in $\mathcal{B}$ 
containing $v$. 
Notice that there exists an optimal coloring 
$c$ such that 
$c(v) \in [n]$ for every vertex $v \in V$ by Observation \ref{obs:maxdeg}.

We regard $G$ as a rooted tree 
with a root $r \in V$.
For each vertex $v \in V$, we denote by $V(v)$ the vertex set of the subtree of $G$ rooted at $v$. 
For each vertex $v \in V$, each integer $k \in [n]$, and each integer $g \in [n]$, we define 
\begin{gather*}    
f(v,k,g)
:=
\min\left\{ \sum_{B \in \mathcal{B}} w(B) \max_{u \in B}c(u)
\,\middle|\,
\begin{array}{l}
c \text{ is a $V(v)$-partial coloring of $G$},\\
\text{$c(u) \leq n$ for all $u \in V$,}\\
c(v) = k,\\
\max\{c(u) \mid u \in B(v)\} = g \\ 
\end{array}
\right \}.
\end{gather*}

We now establish a recursive formula for $f(v,k,g)$.
First, consider the case where $v$ is a leaf of $G$.
For $k \in [n]$ and $g \in [n]$, we obtain
\[
f(v,k,g)=
\begin{cases}
    w(B(v)) g& \text{if $k = g$},\\
    +\infty & \text{otherwise}.\\
\end{cases}
\]

Next, consider the case where $v$ is not a leaf of $G$. 
Let $U$ be the set of children of $v$ and let $W := U \cap B(v)$. 
See \Cref{fig:dp_partition1}.
Notice that
$B(u) = B(v)$ for every $u \in W$.  
Furthermore, since $B(v)$ is connected, 
$B(v) \cap V(u) = \emptyset$
for every $u \in U \setminus W$. 
If $k > g$, then we obtain
$f(v,k,g) = + \infty$ since $c(v) \le
\max\{c(u) \mid u \in B(v)\}$. 
Since $\mathcal{B}$ is 
a connected partition of $V$, 
$|\{u \in U \mid B \cap V(u) \neq \emptyset\}| \le 1$
for each bundle $B \in \mathcal{B} \setminus \{B(v)\}$.
Thus, if $k = g$, then 
\[
\begin{split}
f(v,k,g)
= 
w(B(v))g 
& + 
 \sum_{u \in W} 
\min_{k^{\prime} \in [n] \setminus \{k\},g^{\prime} \in [g]}
\Big(f(u, k^{\prime},g^{\prime})-w(B(u))g^{\prime}\Big) \\
& +
\sum_{u \in U \setminus W}
\min_{k^{\prime} \in [n] \setminus \{k\},g^{\prime} \in [n]}
f(u, k^{\prime},g^{\prime}).
\end{split}
\]
Similarly, if $k < g$, then 
\[
\begin{split}
f(v,k,g)
= 
\min_{u^{\prime} \in W}
& \Big\{
\min_{k^{\prime} \in [n] \setminus \{k\}}f(u^{\prime}, k^{\prime},g) \\
& + 
 \sum_{u \in W \setminus \{u^{\prime}\}} 
\min_{k^{\prime} \in [n] \setminus \{k\},g^{\prime} \in [g]}
\Big(f(u, k^{\prime},g^{\prime})-w(B(u))g^{\prime}\Big) \\
& +
\sum_{x \in U \setminus W}
\min_{k^{\prime} \in [n] \setminus \{k\},g^{\prime} \in [n]}
f(u, k^{\prime},g^{\prime})
\Big\}.
\end{split}
\]

\begin{figure}[t]
    \centering
    \includegraphics[scale=0.8]{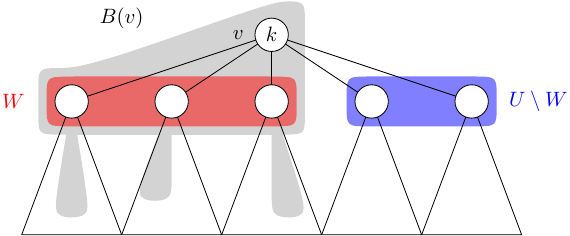}
    \caption{Notation for the case where $\mathcal{B}$ is a partition of $V$.}
    \label{fig:dp_partition1}
\end{figure}

By computing $f(r,k,g)$ over all $k$ and $g$ with these equations, 
we can solve \mscwb in this case. 
For each vertex $v \in V$, we can compute 
$f(v,k,g)$ over all $k$ and $g$ in $O(|\delta_G(v)|^2n^4)$ time. 
Thus, we obtain the following theorem. 
\begin{theorem}
\label{thm:connected_tree}
    \mscwb can be solved in $O(n^6)$ time
    %polynomial time 
    when $G$ is a tree,   
    the subgraph of $G$ induced by each bundle $B \in \mathcal{B}$ is connected, 
    and $\mathcal{B}$ is a partition of $V$. 
\end{theorem}

\subsection{Algorithm when $G$ is a Path with Connected
Bundles}\label{sec:connected_path}

In this section, we present a polynomial-time algorithm for the case when $G=(V,E)$ is a path and the bundle family is connected but not necessarily a partition.
Assume that 
$V = \{v_1,v_2,\dots,v_n\}$ and 
$E = \{v_i v_{i+1} \mid i \in [n-1]\}$. 
In this case, 
there exists an optimal coloring 
$c$ such that 
$c(v) \in [3]$ for every vertex $v \in V$ by Observation \ref{obs:maxdeg}.

For each integer $i \in [n]$, 
we define $V(v_i) := \{v_{i^{\prime}}\mid i^{\prime} \in [i]\}$. 
For each integer $i \in [n]$, each integer $k \in [3]$, and each pair of integers $p,q \in [n] \cup \{0\}$, we define 
\begin{gather*}    
f(i,k,p,q)
\coloneqq
\min\left\{ \sum_{B \in \mathcal{B}} w(B) \max_{u \in B}c(u)
\,\middle|\,
\begin{array}{l}
c \text{ is a $V(v_i)$-partial coloring of $G$},\\
\text{$c(u) \leq 3$ for all $u \in V$,}\\
c(v_i) = k,\\
\max\{i^{\prime} \in [i] \mid c(v_{i^{\prime}}) = 2\} = p,\\
\max\{i^{\prime} \in [i] \mid c(v_{i^{\prime}}) = 3\} = q
\end{array}
\right \},
\end{gather*}
where the maximum over the empty set is defined as $0$.
For each bundle $B \in \mathcal{B}$, we 
define 
${\sf left}(B) \coloneqq \min\{i \in [n] \mid v_i \in B\}$.
For each integer $i \in [n]$, 
define $\mathcal{B}_i$ 
as the set of bundles 
$B \in \mathcal{B}$ such that 
${\sf left}(B) = i$.

For each integer $i \in [n]$, each 
integer $k \in [3]$, and each pair of 
integers $p,q \in [n] \cup \{0\}$, we can 
compute $f(i,k,p,q)$ as follows. 

If $i = 1$, then 
\[
f(i,k,p,q)=
\begin{cases}
    \sum_{B \in \mathcal{B}_i}w(B) & \text{if $k = 1$, $p = q = 0$},\\
    \sum_{B \in \mathcal{B}_i}2w(B) & \text{if $k = 2$, $p = 1$, $q = 0$},\\
    \sum_{B \in \mathcal{B}_i}3w(B) & \text{if $k = 3$, $p = 0$, $q = 1$},\\
    +\infty & \text{otherwise}.\\
\end{cases}
\]

For $i > 1$, we proceed as follows.
First, we consider the case where 
$k = 1$. 
If $p = i$ or $q = i$, then 
$f(i,1,p,q) = +\infty$. 
Assume that 
$p \neq i$ and $q \neq i$. 
We have 
\begin{equation*}
f(i,1,p,q) =
\min_{k^{\prime} \in \{2,3\}}
f(i-1,k^{\prime},p,q)
+
\sum_{B \in \mathcal{B}_i}w(B).
\end{equation*}

Next, we consider the case where 
$k = 2$. 
If $p \neq i$ or $q = i$, then 
$f(i,2,p,q) = +\infty$. 
Assume that 
$p = i$ and $q \neq i$. 
Let $c^{\prime}$ be a $V(v_{i-1})$-partial coloring achieving 
$f(i-1, k^{\prime},p^{\prime},q)$
for some integers $k^{\prime} \in \{1,3\}$ 
and $p^{\prime} \in [i-1] \cup \{0\}$.
Define $c$ as the $V(v_i)$-partial coloring such that 
$c(v_i) = 2$ and $c(v) = c^{\prime}(v)$ for 
every vertex $v \in V(v_{i-1})$.
Then, for each bundle $B \in \mathcal{B}$ such that 
$v_i \in B$ and $\max\{p^{\prime},q\} < {\sf left}(B) \le i-1$, 
we have 
$\max_{u \in B}c^{\prime}(u) = 1$ and 
$\max_{u \in B}c(u) = 2$.
Thus, we have 
\[
\begin{split}
f(i,2,p,q)
= 
\min_{k^{\prime} \in \{1,3\}, p^{\prime} \in [i-1] \cup \{0\}}
& \Big\{
f(i-1, k^{\prime},p^{\prime},q) \\
& + 
\sum_{B \in \mathcal{B} \colon v_i \in B, \max\{p^{\prime},q\} < {\sf left}(B) \le i-1} 
w(B)\Big\} 
+
\sum_{B \in \mathcal{B}_i}2w(B).
\end{split}
\]

Third, we consider the case where 
$k = 3$. 
If $p = i$ or $q \neq i$, then 
$f(i,3,p,q) = +\infty$. 
Assume that 
$p \neq i$ and $q = i$.
In a way similar to the case where 
$k = 2$, we can see that
\[
\begin{split}
f(i,3,p,q)
= 
\min_{k^{\prime} \in \{1,2\}, q^{\prime} \in [i-1] \cup \{0\}}
& \Big\{
f(i-1, k^{\prime},p,q^{\prime}) \\
& + 
\sum_{B \in \mathcal{B} \colon v_i \in B, q^{\prime} < {\sf left}(B) \le p} 
w(B) \\
& + 
\sum_{B \in \mathcal{B} \colon v_i \in B, \max\{p,q^{\prime}\} < {\sf left}(B) \le i-1} 
2w(B)\Big\}
+
\sum_{B \in \mathcal{B}_i}3w(B).
\end{split}
\]

By computing $f(n,k,p,q)$ over all $k$, $p$, and $q$ with these equations, 
we can solve \mscwb in this case. 
Thus, we obtain the following theorem. 
\begin{theorem}
\label{thm:connected_path}
    \mscwb can be solved in $O(n^3|{\cal B}|)$ time when $G$ is a path and    
    the subgraph of $G$ induced by each bundle $B \in \mathcal{B}$ is connected.
\end{theorem}

\section{Bipartite Graphs}

\subsection{NP-Completeness}
In this section, we first prove that \mscwb is NP-complete even when $G$ is a bipartite graph, the number of bundles is four, and the weights are uniform.
To prove NP-completeness, we use the following problem.

\begin{description}
\item[Problem:] \textsc{Bipartite $3$-List-Coloring}
\item[Input:] A bipartite graph $G=(V, E)$ and a list $L(v)\subseteq \{1,2,3\}$ of available colors for each vertex $v\in V$.
\item[Question:] Determine whether there exists a proper coloring $c\colon V\to \{1,2,3\}$ such that $c(v)\in L(v)$ for each $v\in V$.
\end{description}

As observed by Golovach and Paulusma~\cite{DBLP:journals/dam/GolovachP14},
a reduction given by Jansen and Scheffler~\cite{DBLP:journals/dam/JansenS97} proves the NP-completeness of \textsc{Bipartite $3$-List-Coloring}.

\begin{theorem}\label{thm:BipartiteNPhard}
\mscwb is NP-complete even when $G$ is a bipartite graph, $|\mathcal{B}|=4$, and $w(B)=1$ for every $B\in \mathcal{B}$.
\end{theorem}
\begin{proof}
It is easy to see that the problem is in NP\@.
In what follows, we prove that the problem is NP-hard by reducing the problem \textsc{Bipartite $3$-List-Coloring}.

Suppose that we are given an instance of the problem \textsc{Bipartite $3$-List-Coloring}.
Specifically, let $G=(V, E)$ be a bipartite graph, and let $L(v)\subseteq \{1,2,3\}$ be a list of available colors for each vertex $v\in V$.
For $Z\subseteq \{1,2, 3\}$, we define $U_Z = \{v\in V\mid L(v)=Z\}$.
Then $\{U_Z\mid Z\subseteq \{1,2, 3\}\}$ forms a partition of $V$, that is, $V=\bigcup_{Z\subseteq \{1,2, 3\}}U_Z$ and $U_Z\cap U_{Z'}=\emptyset$ for any distinct subsets $Z, Z'\subseteq\{1,2,3\}$.

We construct an instance $(\tilde{G}, \mathcal{B}, w, C)$ of \mscwb, where $\tilde{G}=(\tilde{V}, \tilde{E})$, as follows.
For $Z\in \{\{2\}, \{2,3\}\}$, we denote by $U^{1}_Z$ a copy of $U_Z$, where $u^{1}\in U^{1}_Z$ denotes a copy of $u\in U_Z$.
Similarly, we define two copies $U^{1}_{\{1,3\}}$ and $U^{2}_{\{1,3\}}$ of $U_{\{1,3\}}$, and three copies $U^{1}_{\{3\}}$, $U^{2}_{\{3\}}$, and $U^{3}_{\{3\}}$ of $U_{\{3\}}$.
Then, the vertex set $\tilde{V}$ of $\tilde{G}$ is defined as
\[
\tilde{V}= V\cup U^{1}_{\{2\}}\cup 
\bigcup_{i=1}^3 U^{i}_{\{3\}}\cup
\bigcup_{i=1}^2
 U^{i}_{\{1,3\}}\cup U^{1}_{\{2,3\}}\cup 
\{v_i \mid i \in [16]\}.
\]
The edge set $\tilde{E}$ is defined as
\[
\tilde{E} = E\cup E_{\{2\}}\cup E_{\{3\}}\cup E_{\{1,3\}}\cup E_{\{2,3\}}\cup E_0,
\]
where we define
\begin{align*}
E_{\{2\}} &= \{ uu^1 \mid u\in U_{\{2\}}\},&
E_{\{3\}} &= \{ uu^2,u^2u^1, uu^3 \mid u\in U_{\{3\}}\},\\
E_{\{1,3\}} &= \{ uu^2, u^2u^1 \mid u\in U_{\{1,3\}}\},&
E_{\{2,3\}} &= \{ uu^1 \mid u\in U_{\{2,3\}}\},\\
E_0 &=\rlap{$\{v_1v_2, v_{11}v_{12}\}\cup \{v_iv_{i+1} \mid i\in\{3,4,5\}\cup \{7,8,9\}\cup \{13,14,15\}\}.$}
\end{align*}
We note that each of $E_{\{2\}}$, $E_{\{3\}}$, $E_{\{1,3\}}$, and $E_{\{2,3\}}$ consists of disjoint edges and paths.
Also, $E_0$ consists of $3$ disjoint paths of length $3$ with two disjoint edges.
Since $\tilde{G}$ is obtained from the bipartite graph $G$ by adding edges of $E_{\{2\}}\cup E_{\{3\}}\cup E_{\{1,3\}}\cup E_{\{2,3\}}\cup E_0$, the constructed graph $\tilde{G}$ is bipartite.
Moreover, we define a family of bundles $\mathcal{B}=\{B_1, \dots ,B_4\}$ as
\begin{align*}
B_1 &=U_{\{1\}}\cup U^1_{\{2\}} \cup U^1_{\{3\}}\cup U^3_{\{3\}}\cup U^1_{\{1,3\}}\cup U^1_{\{2,3\}}\cup\{v_{1}, v_{3},v_{6}\},\\
B_2 &=U_{\{2\}}\cup U_{\{1,2\}} \cup U^2_{\{3\}}\cup U^2_{\{1,3\}}\cup\{v_2, v_{7}, v_{10},v_{12}\},\\
B_3 &=U_{\{3\}}\cup U_{\{1,3\}} \cup U_{\{2,3\}}\cup U_{\{1,2,3\}}\cup\{v_{4}, v_{5}, v_{8}, v_{9}, v_{14}, v_{15}\},\\
B_4 &= \{v_{11}, v_{13},v_{16}\}.
\end{align*}
We set $w(B_i)=1$ for $i\in\{1,2,3,4\}$ and $C=7$.
This defines an instance $(\tilde{G}, \mathcal{B}, w, C)$ of \mscwb.
See \Cref{fig:reduction_bipartite_four1}.
By construction, the size of the obtained instance is bounded by a polynomial in the size of $G$.
For a given proper coloring $\tilde{c}$ of $\tilde{G}$, we denote $\gamma_i(\tilde{c}) = \max_{v\in B_i} \tilde{c}(v)$ for each $i\in\{1,2,3,4\}$.
Then, the cost of $\tilde{c}$ is equal to $\sum_{i=1}^4\gamma_i(\tilde{c})$.

Since $\tilde{G}$ is bipartite, $\tilde{G}$ has a proper coloring $\tilde{c}$ such that $\tilde{c}(v)\in\{1,2\}$ for each $v\in \tilde{V}$.
Its cost is at most $8$.
Below, we prove that the instance $(\tilde{G}, \mathcal{B}, w, C)$ of \mscwb is a yes-instance if and only if the given instance $(G, L)$ of \textsc{Bipartite $3$-List-Coloring} is a yes-instance.

We first prove the necessity (i.e., the ``if'' part).
Suppose that there exists a proper coloring $c$ of $G$ such that $c(v)\in L(v)$ for each $v\in V$.
We define a coloring $\tilde{c}$ of $\tilde{G}$  as
\begin{align*}
\tilde{c}(v) =
\begin{cases}
c(v) & \text{if $v\in V$},\\
1 & \text{if $v\in (B_1\setminus V)\cup\{v_{9}\}\cup B_4$},\\
2 & \text{if $v\in (B_2\setminus V)\cup \{v_{5}, v_{15}\}$},\\
3 & \text{if $v\in \{v_{4}, v_{8}, v_{14}\}$}.
\end{cases}
\end{align*}
Then, it is not difficult to see that $\tilde{c}$ is a proper coloring of $\tilde{G}$.
In fact, since $c$ is a proper coloring, the end vertices of each edge of $E$ have different colors.
Moreover, for each edge in $E_{\{2\}}\cup E_{\{3\}}\cup E_{\{1,3\}}\cup E_{\{2,3\}}$, since the end vertices belong to distinct bundles, they are colored differently. In addition, the $3$ disjoint paths in $E_0$ are properly colored.
The cost is equal to $\gamma_1(\tilde{c}) + \gamma_2(\tilde{c}) + \gamma_3(\tilde{c}) + \gamma_4(\tilde{c}) = 1+2+3+1 = 7$.
Thus, the claim holds.

Next, we prove the sufficiency (i.e., the ``only-if'' part).
Suppose that there exists a proper coloring $\tilde{c}$ of $\tilde{G}$ whose cost is at most $7$.
By the definition of a coloring, $\gamma_i(\tilde{c})\geq 1$ for all $i\in [4]$.
Since the cost is at most $7$, there exists at least one bundle $B_j$ such that $\gamma_j(\tilde{c})=1$.
Since there exists an edge between vertices of $B_3$, it follows that $\gamma_3(\tilde{c})\geq 2$, which implies that $j\neq 3$.

The graph $\tilde{G}$ has a path of length $3$ with edges in $E_0$ such that both of its end vertices belong to $B_j$.
Since $\gamma_j(\tilde{c}) = 1$, we need other $2$ colors for the inner vertices of this path.
Since the inner vertices are contained in $B_3$, it holds that $\gamma_3(\tilde{c})\geq 3$, which implies that  $\gamma_1(\tilde{c})+\gamma_2(\tilde{c})+\gamma_4(\tilde{c}) \leq 7 - \gamma_3(\tilde{c}) \leq 4$.
Since the two edges $v_1v_2$ and $v_{11}v_{12}$ satisfy $v_1\in B_1$, $v_2, v_{12}\in B_2$, and $v_{11}\in B_4$,
we see that $(\gamma_1(\tilde{c}), \gamma_2(\tilde{c}), \gamma_4(\tilde{c}))=(1,2,1)$ is the only possible combination.
Therefore, we obtain  
$(\gamma_1(\tilde{c}), \gamma_2(\tilde{c}), \gamma_3(\tilde{c}),  \gamma_4(\tilde{c}))= (1,2,3,1)$.

We claim that, for any vertex $v\in V$, it holds that $\tilde{c}(v)\in L(v)$, which means that a coloring obtained from $\tilde{c}$ by restricting on $V$ is a solution to the instance $(G, L)$ of \textsc{Bipartite $3$-List-Coloring}.

First, suppose that $v\in B_1\cap V$, i.e., $L(v)=\{1\}$.
Then, since $v\in B_1$, $\tilde{c}(v)\leq \gamma_1(\tilde{c})=1$.
Thus, $\tilde{c}(v)\in L(v)$.
Next, suppose that $v\in B_2\cap V$, that is, $L(v)=\{2\}$ or $L(v)=\{1,2\}$.
In this case, we have $\tilde{c}(v)\leq 2$.
If $L(v)=\{2\}$, then there exists a vertex $u\in U^1_{\{2\}}\subseteq B_1$ such that $uv\in \tilde{E}$.
Since $\tilde{c}(u)=1$, we have $\tilde{c}(v)=2\in L(v)$.
Thus, for any vertex $v\in B_2\cap V$, we have $\tilde{c}(v)\in L(v)$.
Finally, suppose that $v\in B_3\cap V$, that is, $L(v)$ is one of $\{3\}$, $\{1, 3\}$, $\{2, 3\}$, and $\{1, 2, 3\}$.
Then, it holds that $\tilde{c}(v)\leq 3$.
If $L(v)=\{3\}$, then there exist $3$ edges $vv^2$, $v^2v^1$, and $vv^3$ where $v^1, v^3\in B_1$ and $v^2\in B_2$, which implies that $\tilde{c}(v)=3\in L(v)$, as $\tilde{c}(v^1)=\tilde{c}(v^3)=1$ and $\tilde{c}(v^2)\leq 2$.
Similarly, we can argue that $\tilde{c}(v)\in L(v)$ for the other cases.
Therefore, for any vertex $v\in B_3\cap V$, we have $\tilde{c}(v)\in L(v)$.

Thus, the theorem holds.
\end{proof}

For the weighted case, \mscwb on bipartite graphs is NP-complete even when $|\mathcal{B}|=3$.
The proof is an adaptation of that for \Cref{thm:BipartiteNPhard}.

\begin{theorem}\label{thm:BipartiteNPhard2}
\mscwb is NP-complete even when $G$ is a bipartite graph, $|\mathcal{B}|=3$, and $w(B)\in\{1,2\}$ for every $B\in \mathcal{B}$.
\end{theorem}
\begin{proof}
We reduce the problem \textsc{Bipartite $3$-List-Coloring}, where the reduction is similar to the one for Theorem~\ref{thm:BipartiteNPhard}.
Suppose that we are given an instance $(G, L)$ of \textsc{Bipartite $3$-List-Coloring}, where $G=(V, E)$ is a bipartite graph and $L(v)\subseteq \{1,2,3\}$ is a list of available colors for each vertex $v\in V$.
In the proof of Theorem~\ref{thm:BipartiteNPhard}, we constructed an instance $(\tilde{G}, \mathcal{B}, w, C)$ of \mscwb.
We now modify $(\tilde{G}, \mathcal{B}, w, C)$ to obtain an instance $(\tilde{G}', \mathcal{B}', w', C')$, where $\tilde{G}'=(\tilde{V}', \tilde{E}')$ is obtained from $\tilde{G}$ by removing the vertices $v_{11},\dots, v_{16}$.
A family of bundles $\mathcal{B}'=\{B'_1, B'_2 ,B'_3\}$ is defined as
\[
    B'_1 =B_1,\quad
    B'_2 =B_2\setminus \{v_{12}\},\quad
    B'_3 =B_3\setminus \{v_{14}, v_{15}\}.
\]
We set $w'(B'_1)=2$, $w'(B'_2)=w'(B'_3)=1$, and $C'=7$.
This defines an instance $(\tilde{G}', \mathcal{B}', w', C')$ of \mscwb.
Then, the argument analogous to the proof of Theorem~\ref{thm:BipartiteNPhard} works.
\end{proof}

\subsection{Tractable Cases}

To complement Theorems~\ref{thm:BipartiteNPhard} and~\ref{thm:BipartiteNPhard2} concerning bipartite graphs, we prove that \mscwb can be solved in polynomial time under two specific conditions: when $|\mathcal{B}|\leq 2$ and when $|\mathcal{B}|\leq 3$ with uniform weights.

\begin{theorem}
\label{thm:poly-bip-gen-2}
\mscwb can be solved in polynomial time when $G$ is bipartite and $|\mathcal{B}|\leq 2$.
\end{theorem}
\begin{proof}
Let $\mathcal{B}=\{B_1, B_2\}$.
We denote $w_i = w(B_i)$ for $i\in\{1,2\}$.
The optimal cost is at least $w_1+ w_2$.
%and moreover,  the optimal value is $w_1+ w_2$ if and only if $G$ is a bipartite graph with partition $B_1$ and $B_2$, which can be checked in polynomial time.
Since $G$ is bipartite, $G$ has a proper coloring with $2$ colors $\{1,2\}$.
Hence, the optimal cost is at most $2w_1+2w_2$.

%We prove that we can check in polynomial time whether there exists a coloring with cost less than $2w_1+2w_2$.
If the optimal cost is less than $2w_1+2w_2$, then $G[B_1]$ or $G[B_2]$ must be colored with color~$1$ only.
%In this case, it is necessary that $B_1$ is a stable set of $G$.
Suppose that $G[B_1]$ is $1$-colorable (i.e., $B_1$ is independent).
Then, by extending the coloring of $G[B_1]$, we can find a coloring $c_1$ of $G[B_2]$ such that $\max_{v\in B_2}c_1(v)$ is minimized since a $2$-coloring of a connected bipartite graph is unique up to permutations of colors.
The cost in this case is $w_1+\gamma_2 w_2$ where $\gamma_2=\max_{v\in B_2}c_1(v)$.
Similarly, if $G[B_2]$ is colored with color $1$, then we can find a coloring $c_2$ of $G$ such that $G[B_2]$ is colored with one color and $\max_{v\in B_2}c_2(v)$ is minimized.
The cost in this case is $\gamma_1 w_1+w_2$ where $\gamma_1=\max_{v\in B_2}c_2(v)$.
Then, we observe that the optimal cost is equal to $\min\{2w_1+2w_2, w_1+\gamma_2 w_2, \gamma_1 w_1+w_2\}$.
Thus, we can find an optimal coloring for a given instance.
\end{proof}

\begin{theorem}
\label{thm:poly-bip-unif-3}
\mscwb with uniform weight can be solved in polynomial time when $G$ is bipartite and $|\mathcal{B}|\leq 3$.
\end{theorem}
\begin{proof}
Let $\mathcal{B}=\{B_1, B_2, B_3\}$.
%Given a coloring $c$, we denote $\gamma_i = \max_{v\in B_i} c(v)$ for $i\in\{1,2,3\}$.
%For each coloring,  $\gamma_i\geq 1$ for each $i\in\{1,2,3\}$.
Since $G$ is bipartite, $G$ has a proper coloring with $2$ colors.
Hence, the optimal cost is at most $2+2+2=6$.

If there exists a coloring with cost less than $6$, then some bundle must be colored with color $1$ only.
We guess a set of bundles that are colored with color $1$ only.
Then, we can extend in polynomial time the coloring to a coloring of $G$ with minimum cost.
Therefore, we can decide whether there exists an optimal coloring if the optimal cost is less than $6$.
Thus, the theorem holds.
\end{proof}

\section{Concluding Remarks}

We have completed a detailed complexity analysis of \mscwb for trees and bipartite graphs, and solved an open problem by Darbouy and Friggstad~\cite{DBLP:conf/swat/DarbouyF24}.
Here, we pose a few open problems.
First, we do not know the existence of an FPT algorithm for trees when $|V|-|\mathcal{B}|$ is a parameter and the bundles form a partition of $V$.
Second, we do not know the complexity for bipartite graphs when the bundles form an independent partition or a connected partition.
Third, we do not know any better approximation ratio for trees and bipartite graphs than the one proposed by Darbouy and Friggstad~\cite{DBLP:conf/swat/DarbouyF24}.

%%
%% Bibliography
%%

%% Please use bibtex, 

% \clearpage 

\bibliography{coloring_with_bundles}

%\clearpage 

\appendix

\section{Treewidths and Courcelle's theorem}

Let $G=(V,E)$ be a graph.
A \emph{tree decomposition} of $G$ is a tree $\mathcal{T}=(\mathcal{V}, \mathcal{E})$ satisfying the following four conditions.
\begin{enumerate}
    \item Each vertex $B \in \mathcal{V}$ is a subset of $V$. To distinguish from a vertex of $G$, a vertex of $\mathcal{T}$ is called a \emph{bag}.
    \item Each vertex $v \in V$ of $G$ is contained in at least one bag $B \in \mathcal{V}$.
    \item For each edge $uv \in E$, there exists at least one bag $B \in \mathcal{V}$ containing both $u$ and $v$.
    \item For each vertex $v \in V$, the bags containing $v$ induce a connected subgraph of $\mathcal{T}$.
\end{enumerate}
The maximum size of a bag minus one is called the \emph{width} of a tree decomposition.
The minimum width of a tree decomposition of $G$ is called the \emph{treewidth} of $G$.
Trees have treewidth one and series-parallel graphs have treewidth two.
It is known that a tree decomposition of $G$ of width $w$ can be found in fixed-parameter tractable time when $w$ is a parameter~\cite{DBLP:journals/siamcomp/Bodlaender96}.

Courcelle's theorem states that every property of graphs $G$ expressible by a formula $\varphi$ in monadic second-order logic can be decided in fixed-parameter-tractable time when the treewidth of $G$ and the length of $\varphi$ are parameters.

In monadic second-order logic of undirected graphs, we can use the usual logical connectives $\neg, \wedge, \vee, \rightarrow, \leftrightarrow, (, )$ and quantifiers $\forall, \exists$.
To form a formula, we are given an undirected graph $G=(V,E)$, and possibly a finite family $\{V_1,V_2,\dots,V_s\}$ of vertices and a finite family $\{E_1,E_2,\dots,E_t\}$ of edges.
Then, the variables are vertices $v \in V$, edges $e \in E$, subsets $X \subseteq V$ of vertices and subsets $Y \subseteq E$ of edges.
The predicates are those in the following forms:
\begin{itemize}
  \item $v \in V_i$ for a vertex variable $v$ and $i \in [s]$;
  \item $e \in E_j$ for an edge variable $e$ and $j \in [t]$;
  \item $v \in X$ for a vertex variable $v$ and a vertex subset variable $X$;
  \item $e \in Y$ for an edge variable $e$ and an edge subset variable $Y$;
  \item $v \in e$ for a vertex variable $v$ and an edge variable $e$;
  \item $uv \in E$ for vertex variables $u$ and $v$.
\end{itemize}
The quantification by $\forall$ and $\exists$ is allowed over vertex variables, edge variables, vertex subset variables and edge subset variables.\footnote{%
More precisely, when we do not allow quantification over edge subset variables, the logic is called MSO${}_1$. On the contrary, when we allow quantification over edge subset variables the logic is called MSO${}_2$. Our definition here refers to MSO${}_2$, but in the proof of our theorem, we only need MSO${}_1$.
}

Then, Courcelle's theorem can be stated as follows.
\begin{theorem}[\cite{DBLP:journals/jal/ArnborgLS91,DBLP:journals/iandc/Courcelle90,DBLP:journals/ita/Courcelle92}]
For a finite undirected graph $G=(V,E)$, a finite family $\mathcal{F}_V = \{V_1,V_2,\dots,V_s\}$ of vertex subsets, a finite family $\mathcal{F}_E = \{E_1,E_2,\dots,E_t\}$ of edge subsets, and a formula $\varphi$ in monadic second-order logic of graphs, given a tree decomposition of $G$ of width at most $w$, we can decide whether the tuple $(G, \mathcal{F}_V, \mathcal{F}_E)$ satisfies the property expressed by $\varphi$ in $O(f(w,|\varphi|)|V|)$ time, where $f\colon \mathbb{Z}_{\geq 0}\times \mathbb{Z}_{\geq 0} \to \mathbb{Z}_{\geq 0}$ is some computable function and $|\varphi|$ is the length of $\varphi$.
\end{theorem}
We note that the sizes of $\mathcal{F}_V$ and $\mathcal{F}_E$ are taken into account in $|\varphi|$.

\end{document}